\documentclass{amsart}

\newtheorem{theorem}{Theorem}[section]
\newtheorem{corollary}{Corollary}[section]
\newtheorem{lemma}[theorem]{Lemma}

\theoremstyle{definition}
\newtheorem{definition}[theorem]{Definition}

\theoremstyle{remark}
\newtheorem{remark}[theorem]{Remark}
\numberwithin{equation}{section}

\newcommand{\A}{\mathcal{A}}

\begin{document}
\title[On the dimensions of oscillator algebras]{ON THE  DIMENSIONS OF  OSCILLATOR ALGEBRAS INDUCED BY ORTHOGONAL POLYNOMIALS}
\author{G. Honnouvo$^1$ and K. Thirulogasanthar$^2$}
\address{$^1$ Department of Mathematics and Statistics, McGill University,  
   805 Sherbrooke Street West, Montreal, Quebec H3A 2K6, Canada }
\email{g\_honnouvo@yahoo.fr}
\address{$^2$Department of Computer Science and Software Engineering, Concordia University,  1455 de Maisonneuve Blvd. West, Montreal, Quebec, H3G 1M8, Canada }
\email{santhar@gmail.com}
\date{\today}
\begin{abstract}
There is a generalized oscillator algebra associated with every class of orthogonal polynomials $\{\Psi_n(x)\}_{n=0}^{\infty}$, on the real line, satisfying a three term recurrence relation $x\Psi_n(x)=b_n\Psi_{n+1}(x)+b_{n-1}\Psi_{n-1}(x),~\Psi_0(x)=1,~b_{-1}=0$.
This note presents necessary and sufficient conditions on $b_n$ for such algebras to be of finite dimension. As examples, we discuss the dimensions of oscillator algebras associated with Hermite, Legendre and Gegenbauer polynomials. Some remarks on the dimensions of oscillator algebras associated with multi-boson systems are also presented.
\end{abstract}
\maketitle
\pagestyle{myheadings}
\large
\section{Introduction}
The connection between classical orthogonal polynomials and the classical groups as well as  the quantum groups  is well known \cite{key1, key2, FV}.  It is also well known that $e^{-x^2/2}H_n(x)$, where $H_n(x)$ are the Hermite polynomials,  are the eigensolutions of the quantum harmonic oscillator Hamiltonian \cite{G}. The connection of orthogonal polynomials with Heisenberg algebra of generalized oscillators is discussed, for example, in \cite{key3, key4, key5,BH}.
 In \cite{keyb},  a  preassigned Hilbert space with an orthogonal polynomial system as a basis is considered as a Fock space. As usual, the ladder operators (annihilation) $a^-$ and (creation) $a^+$ as well as the number operator $N$ are defined in this space. Using a standard procedure the author used these operators to build the selfadjoint operators; the position operator $X$, the momentum operator $P$ as well as the energy operator $H=X^2 + P^2.$  By analogy with the usual Heisenberg algebra these operators generate an algebra, which  is called a generalized oscillator algebra. Further, the operator $H$ has a simple discrete spectrum and the initial orthogonal polynomial system is an eigenfunction system of $H$. As applications, using this algebra, spectrum and the polynomials one can also obtain coherent states of quantum optics \cite{BD, BH}.

In this note we shall consider generalized oscillator algebras associated with orthogonal polynomials $\{\Psi_n(x)\}_{n=0}^{\infty}$, on the real line, satisfying a three term recurrence relation $x\Psi_n(x)=b_n\Psi_{n+1}(x)+b_{n-1}\Psi_{n-1}(x),~\Psi_0(x)=1,~b_{-1}=0$. In fact we shall provide necessary and sufficient condition on $b_n$ for these algebras to be of finite dimension. As an application, we shall also discuss the dimensions of algebras associated with multi-boson systems.  This note, in a manner, provides a dimension-wise classification for the algebras generated by orthogonal polynomials satisfying three term recurrence relations of the form $x\Psi_n(x)=b_n\Psi_{n+1}(x)+b_{n-1}\Psi_{n-1}(x),~\Psi_0(x)=1,~b_{-1}=0$.

\section{Preliminaries and formulation of the problem}
In this section, we shall extract the preliminary materials as needed here from \cite{keyb}. Let $\mu$ be a positive Borel measure on the real line, $\mathbb{R}$, such that
\begin{equation*}
\int_{-\infty}^{\infty} \mu(dx)=1,\quad \text{and}\quad \mu_{2k+1}=\int_{-\infty}^{\infty}x^{2k+1} \mu(dx)=0;\quad k=0,1,...
\end{equation*}
The measure $\mu$ is called a symmetric probability measure. By $\mathfrak{H}$ we denote the Hilbert space $L^2(\mathbb{R}, \mu ).$
 Let 
\begin{equation}\label{sb}
\{ b_n \}_{n=0}^\infty,\quad b_n>0; \quad n=0, 1,... 
\end{equation}
be a positive sequence defined by the algebraic equations system
\begin{equation}\label{a-1}
\sum_{m=0}^{[\frac{1}{2}n]}\sum_{s=0}^{[\frac{1}{2}n]}(-1)^{m+s}\alpha_{2m-1,n-1}\alpha_{2s-1,n-1}\frac{\mu_{2n-2m-2s+2}}{(b_{n-1}^2)!}= b_{n-1}^2 +b_{n}^2;\quad n=1,2,...,
\end{equation}
where $(b_{n-1}^2)!= b_{0}^2b_{1}^2...b_{n-1}^2,$ the integral part of $a$ is denoted by $[a]$, and the coefficients $\alpha_{ij}$ are given by
\begin{eqnarray}
\alpha_{2p-1,n-1}= \sum_{k_1=2p-1}^{n-1}b_{k_1}^2\sum_{k_2=2p-3}^{k_1-2}b_{k_2}^2
...\sum_{k_p=1}^{k_{p-1}-2}b_{k_p}^2.
\end{eqnarray}
Let us consider a system $\{\Psi_n(x) \}_{n=0}^\infty$ of polynomials defined by the recurrence relations $(n\geq 0):$
\begin{equation}\label{Re}
x\Psi_n(x)=b_n\Psi_{n+1}(x)+ b_{n-1}\Psi_{n-1}(x),\quad \Psi_0(x)=1,\quad b_{-1}=0,
\end{equation}
where $\{b_n \}_{n=0}^\infty$ is a given positive sequence satisfying the relation $(\ref{a-1})$.
The following theorem was proved in \cite{keyb}.
\begin{theorem}
The polynomial system $\{\Psi_n(x) \}_{n=0}^\infty$ is orthonormal in the Hilbert space $\mathfrak{H}$ if and only if the coefficients $b_n$ and the moments $\mu_{2k}$ are connected by relation (\ref{a-1}).
\end{theorem}

\subsection{The generalized oscillator algebra}
Let $\{\Psi_n(x) \}_{n=0}^\infty$ be an orthonormal basis of the Fock space $\mathcal H$ which satisfies the recurrence relation (\ref{Re}). That is
$$\mathcal{H}=\overline{\text{span}}\left\{ \Psi_n(x)~|~n=0,1,2,...\right\},$$
where the bar stands for the closure of the linear span. In this subsection, we present the definition of the generalized oscillator algebra corresponding to the system $\{\Psi_n(x) \}_{n=0}^\infty$. Define the ladder operators $A^\dagger$ and $A$ in the Fock space, $\mathcal{H}$, by the usual formulas:
\begin{equation}\label{b-1}
A^\dagger\Psi_n(x)=\sqrt{2}b_n \Psi_{n+1}(x),\quad A\Psi_n(x)=\sqrt{2}b_{n-1} \Psi_{n-1}(x).
 \end{equation}

\begin{definition}
An operator $N$ in the Fock space $\mathcal{H}$ equipped with the basis $\{\Psi_n(x) \}_{n=0}^\infty$ is called a number operator if it acts on the basis vectors as 
$N\Psi_n(x)=n\Psi_n(x), \quad n\geq 0.$
\end{definition}
\begin{remark}\label{Rem1}
 A word about the symmetry of the measure is in order.
 A polynomial set$\{\Psi_n(x) \}_{n=0}^\infty$ is called a canonical polynomial system if it is defined by the recurrence relation (\ref{Re}). The canonical polynomial system $\{\Psi_n(x) \}_{n=0}^\infty$ is uniquely determined by the symmetric probability measure $\mu$.
Now, as usual, let the position operator be
$$
Q=\frac{A+A^{\dagger}}{\sqrt{2}}.
$$
In \cite{keyb} it is argued that, under the above operator set up, when the measure is non-symmetric the position operator is not necessarily be an operator of multiplication by an independent variable. In this regard, in order to guarantee
\begin{equation}\label{pos}
Q\Psi_n(x)=x\Psi_n(x)
\end{equation}
the symmetry of the measure is required  and the relation (\ref{pos}) is essential for the three term recursion relation (\ref{Re}). In order to see this, consider
\begin{eqnarray*}
Q\Psi_n(x)&=&x\Psi_n(x)\\
&=&\frac{1}{\sqrt{2}}(A+A^{\dagger})\Psi_n(x)\\
&=&\frac{1}{\sqrt{2}}(\sqrt{2}b_{n-1}\Psi_{n-1}(x)+\sqrt{2}b_n\Psi_{n+1})\\
&=&b_{n-1}\Psi_{n-1}(x)+b_n\Psi_{n+1}(x).
\end{eqnarray*}
Further, the above computation describes the relation between the operators $A$ and $A^{\dagger}$ and the three term recursion relation (\ref{Re}). It may also be interesting to note that for a symmetric measure one can reduce a non-symmetric recursion relation to a symmetric one. The symmetric and non-symmetric schemes are treated separately in \cite{keyb}.
\end{remark}
The operator $N$ is self adjoint in the Fock space. Therefore for any Borel function $B$, through the spectral theorem \cite{Oli}, one can define the operator $B(N)$. In this regard, we take a function
 $B(N)$ of operator $N$ in the space $\mathcal{H}$ which acts on the basis vectors, $\{\Psi_n(x) \}_{n=0}^\infty$ as 
\begin{equation}\label{b3}
B(N)\Psi_n(x)=b_{n-1}^2\Psi_n(x),\quad\text{and}\quad B(N+I)\Psi_n(x)=b_n^2\Psi_n(x);\quad n\geq 0.
\end{equation}
However, our main result shall be obtained by  restricting the function $B$ to be polynomial of the operator $N$.
The following result is proved in \cite{keyb, BD}:
\begin{theorem}\label{tx1}
The operators $A, A^{\dagger}$ and $N$ obey the following commutation relations
\begin{equation}\label{w}
[A, A^\dagger]= 2\left(B(N+I)-B(N)\right), \quad [N, A^\dagger]=A^\dagger,\quad [N, A]=-A.
\end{equation}
Moreover if there is a real number $G$ and a real function $C(n)$, such that 
\begin{equation}
b_n^2-Gb_{n-1}^2= C(n), \quad n\geq 0, \quad b_{-1}=0,
\end{equation}
then the operators (\ref{b-1}), (\ref{b3}) fulfil the relation
\begin{equation}
A A^\dagger-GA^\dagger A=2C(N).
\end{equation}
The function $C(N)$ is defined similar to (\ref{b3}) with $C(n)$ instead of $b_{n-1}^2.$
\end{theorem}

\begin{definition}\label{DD1}
An algebra $\mathcal A$ is called a generalized oscillator algebra corresponding to the orthonormal system $\{\Psi_n(x) \}_{n=0}^\infty$, which satisfies (\ref{Re}), if $\mathcal A$ is generated by the operators $A^\dagger$, $A$, $N$, and $I$ and by their commutators. These operators should also satisfy the relations (\ref{b-1}) and (\ref{w}).  
\end{definition}
\begin{remark}
At this point we like to emphasize a word about the above definition.  The algebra $\mathcal{A}$ consists the operators $A, A^{\dagger}, N$ and $I$ and their repeated commutators only.
\end{remark}
\begin{remark}
The above oscillator algebra arises only in the symmetric case. For the non-symmetric case one can also construct generalized oscillator algebra. However the oscillator Hamiltonian takes the standard form only in a few {\em coordinate-impulse} operators, which can result from the operators $X$ (position) and $P$ (momentum) by rotation \cite{keyb}.
\end{remark}
\section{Main results}
In this section we prove the main result of the paper. That is, we prove the necessary and sufficient condition, in terms of $b_n$, for the finiteness of the dimension of the generalized oscillator algebra defined in Definition (\ref{DD1}). For this, let us define the following sequence of operators:
\begin{equation}
{M}_0= [A, A^{\dagger}], \quad {M}_{1}^+= [{M}_0, A^{\dagger}],\cdots, {M}_{j}^+= [{M}_{j-1}^+, A^{\dagger}]; \quad j=2,3,...\nonumber
 \end{equation}
The following are the main results:
\begin{theorem}\label{t1}
The generalized oscillator algebra, defined in Definition (\ref{DD1}), is of finite dimension if and only if $b_n^2$ has the following form
\begin{equation}\label{T1}
b_n^2= R(n),
 \end{equation}
where R is a classical polynomial uniquely determined by $b_n$,
\begin{equation}
R(n)=a_0+a_1n+a_2n^2;\quad a_0,a_1,a_2 \in \mathbb{R}.
 \end{equation}
\end{theorem}
As a corollary we can state the following result.
\begin{corollary}\label{c1}
If the oscillator algebra defined in Definition (\ref{DD1}) is of finite dimension, then the dimension of the algebra is four.
\end{corollary}
In order to give a proof to Theorem \ref{t1}, we need to prove few preliminary results and these results appear in the following lemmas.
For this, let us consider the following two variables sequence $\{\mathfrak{V}_n^{(j)} \}_{n=0,j=1}^\infty$ defined by 
\begin{equation}
 \mathfrak{V}_n^{(0)}= b_n^2-b_{n-1}^2,\cdots, \mathfrak{V}_n^{(j)}=\mathfrak{V}_{n+1}^{(j-1)}-\mathfrak{V}_{n}^{(j-1)};\quad j=1,2,...,\quad  
n=0,1,....\nonumber
\end{equation}
\begin{lemma}\label{L1}
If  for every fixed $j>0$~~ 
$\mathfrak{V}_n^{(j)}\neq\text{constant}, n=0,1,....$, then the generalized oscillator algebra is of infinite dimension.
\end{lemma}
\begin{proof}It is easy to show that
\begin{eqnarray*}
{M}_{0}\Psi_n&=&2\mathfrak{V}_n^{(0)}\Psi_{n}\\
{M}_{1}^+\Psi_n&=&2(\sqrt{2})b_{n}\mathfrak{V}_n^{(1)}\Psi_{n+1}\\
{M}_{2}^+\Psi_n&=&2(\sqrt{2})^2b_{n}b_{n+1}\mathfrak{V}_n^{(2)}\Psi_{n+2}\\
&&\vdots\;\;\;\;\;\;\vdots\;\;\;\vdots\\
{M}_{j}^+\Psi_n&=&2(\sqrt{2})^j\left(\prod_{i=0}^{j-1}b_{n+i}\right)\mathfrak{V}_n^{(j)}\Psi_{n+j}.
\end{eqnarray*}
Therefore if for every $j$, $\mathfrak{V}_n^{(j)}\neq\text{constant};\quad n=0,1,....,$ then ${M}_{j}^+$ is a new operator in the algebra. That is, since  ${M}_{j}^+ \neq {M}_{k}^+$ for $j\neq k$ the algebra is of infinite dimension.
\end{proof}
\begin{lemma}\label{L2}
Suppose that there exists $j>0$ such that 
$$\mathfrak{V}_n^{(j)}=\text{constant};\quad n=0,1,...$$ 
Let $\mathbb S$ be the set of $j$ satisfying the above property. 
Let $j_0= \inf \mathbb S$.
Then there exists a classical polynomial $P(n)$ of degree $j_0$ such that 
\begin{equation}
b_n^2=\sum_{i=1}^{n}P(i)+b_0^2,\quad\text{where}
\end{equation} 
\begin{equation}
P(n)= \sum_{i=0}^{j_0}\alpha_in^i;\quad n=0,1,...,\quad \alpha_i\in \mathbb{R}.
\end{equation}
\end{lemma}
\begin{proof}
Since $\mathfrak{V}_{n+1}^{(j_0-1)}-\mathfrak{V}_n^{(j_0-1)}=  \mathfrak{V}_n^{(j_0)}= constant;\quad n=0,1,....,$ we have the following iteration:
\begin{eqnarray*}
&&\mathfrak{V}_{n+1}^{(j_0-1)}-\mathfrak{V}_n^{(j_0-1)}=\mathfrak{V}_n^{(j_0)}= constant\\
&&\mathfrak{V}_n^{(j_0-1)}= a_1n+a_2\\
&&\mathfrak{V}_{n+1}^{(j_0-2)}-\mathfrak{V}_n^{(j_0-2)}=\mathfrak{V}_n^{(j_0-1)}= a_1n+a_2\\
&&\mathfrak{V}_n^{(j_0-2)}=c_1n^2+c_2n+c_3\\
&&\vdots\;\;\;\;\;\;\vdots\;\;\;\vdots\\
&&\mathfrak{V}_n^{(0)}  = b_n^2-b_{n-1}^2 =P(n)= \sum_{i=0}^{j_0}\alpha_in^i,\quad n=0,1,....
\end{eqnarray*}
which implies that 
\begin{equation}\label{D}
b_n^2=\sum_{i=1}^{n}P(i)+b_0^2
\end{equation}
\end{proof}
\begin{lemma}\label{L3}
There exists a polynomial $Q$ of degree $j_0+1$ such that
\begin{equation}\label{E}
Q(n)-Q(n-1)= P(n);\quad n=1,2,...
\end{equation}
\end{lemma}
\begin{proof}
In order to prove (\ref{E})
let 
$$Q(n)=\sum_{i=0}^{j_0+1}a_in^i;\quad a_i \in \mathbb{R},\quad i=0,1,...,j_0+1.$$
Thereby, the equation (\ref{E}) is 
$$\sum_{i=1}^{j_0+1}a_in^i-\sum_{i=1}^{j_0+1}a_i\left(\sum_{k=0}^{i}\left(\begin{array}{c}i\\ k\end{array}\right)n^k(-1)^{i-k}\right)=\sum_{i=0}^{j_0}\alpha_in^i,$$
which is equivalent to the linear system
\begin{equation}\label{eq}
AX=Y,
\end{equation}
where the $(j_0+1)\times(j_0+1)$ matrix $A$ is given by
\begin{equation*}
A=\left(\begin{array}{cccccc}
C_{j_{0}+1}^{j_{0}}(-1)^{}&0&0&0&\hdots &0\\
C_{j_{0}+1}^{j_{0}-1}(-1)^{2}&C_{j_{0}}^{j_{0}-1}(-1)^{}&0&0&\hdots &0\\
C_{j_{0}+1}^{j_{0}-2}(-1)^{3}&C_{j_{0}}^{j_{0}-2}(-1)^{2}&C_{j_{0}-1}^{j_{0}-2}(-1)^{}&0&\hdots &0\\
\vdots&\vdots&\vdots&\vdots&\vdots&\vdots\\
\vdots&\vdots&\vdots&\vdots&\vdots&\vdots\\
\vdots&\vdots&\vdots&\vdots&\vdots&\vdots\\
C_{j_{0}+1}^{1}(-1)^{j_0}&C_{j_{0}}^{1}(-1)^{j_0-1}&C_{j_{0}-1}^{1}(-1)^{j_0-2}&\hdots &C_{2}^{1}(-1)&0\\
C_{j_{0}+1}^{0}(-1)^{j_0+1}&C_{j_{0}}^{0}(-1)^{j_0}&C_{j_{0}-1}^{0}(-1)^{j_0-2}&\hdots &\hdots&C_{1}^{0}(-1)
\end{array}
\right),
\end{equation*}
\begin{eqnarray*}
X&=&\left(\begin{array}{cccccccc}
a_{j_0+1}&
a_{j_0}&
\hdots&
\hdots&
\hdots&
\hdots&
a_{2}&
a_{1}\end{array}\right)^{T}\;\;\text{and}\\ 
Y&=&\left(\begin{array}{cccccccc}
\alpha_{j_0}&
\alpha_{j_0-1}&
\hdots&
\hdots&
\hdots&
\hdots&
\alpha_{1}&
\alpha_{0}\end{array}\right)^{T},
\end{eqnarray*}
and where $C_k^l=-\left(\begin{array}{c}k\\ l\end{array}\right)$, the binomial coefficient. Since the determinant, $\det(A)=\frac{1}{2}(j_0+1)(j_0+2)\neq 0$, the reals $a_i,\quad i=1,2,...,j_0+1$ are uniquely determined and $a_0$ can take any real value. This proves the lemma.
\end{proof}
\begin{remark}\label{Rem2}
Since the algebras corresponding to $b_n^2=\sum_{i=0}^{p}\theta_i n^i $ ($\theta_i\in\mathbb{R}$) and
$b_n^2=n^p+\sum_{i=0}^{p-1}\gamma_i n^i $ are isomorphic, we shall use either one as convenient.
\end{remark}
\subsection{Proof of Theorem \ref{t1}}
Suppose the algebra is of finite dimension. Then by Lemma \ref{L1} we have
$$\mathfrak{V}_n^{(j)}=constant;\quad n=0,1,2,....;\quad j>0.$$
Thereby in Lemma \ref{L2}, $j_0=1$ and
$$b_n^2=\sum_{i=1}^nP(i)+b_0^2\quad\text{with}\quad P(i)=\sum_{k=0}^1\alpha_ki^k;\quad n=0,1,2,...$$
Thus
\begin{eqnarray*}
b_n^2&=&\sum_{i=1}^{n}(\alpha_0+\alpha_1i)+b_0^2\\
&=&\alpha_0n+\alpha_1\frac{n(n+1)}{2}+b_0^2\\
&=&\frac{\alpha_1}{2}n^2+(\alpha_0+\frac{\alpha_1}{2})n+b_0^2,
\end{eqnarray*}
which is a second degree polynomial in $n$ and choices of $\alpha$'s can reduce it to a first or zero degree polynomial as well.\\

To prove the converse part, let $b_n^2$ have the following form,
\begin{equation}
b_n^2=Q(n)=\sum_{i=0}^{p}\theta_in^i;\quad \theta_i \in \mathbb{R},\quad i=0,1,...,p
 \end{equation} 
and let $\mathcal{A}(p)$ denotes the corresponding algebra. We prove the converse by cases.\\
{\bf The case p=0:}
If $b_n^2=\text{constant}\; \neq 0\;\:\;n=0,1,2,... $,  then we have $$[N,A^{\dagger}]=A^{\dagger},\;\;[N,A]=-A,\:\:[A,A^{\dagger}]=0.$$
So, the algebra is generated by the set of operators $\{ A^+, A, N, I  \}$ is of finite dimension, and  the dimension of $\mathcal{A}(0)$ is four.\\
{\bf The case p=1:}
If $b_n^2=n+\alpha$, then we have
 $$[N,A^{\dagger}]=A^{\dagger},\;\;[N,A]=-A,\:\:[A,A^{\dagger}]=2I,$$
the set of generators of the corresponding algebra is $\{ A^+, A, N, I  \}$. Therefore, the dimension of $\mathcal{A}(1)$ is four. In fact, this algebra is isomorphic to the harmonic oscillator  algebra which is obtained when $b_n^2= \frac{n+1}{2}.$\\
{\bf The case p=2:}
If
$b_{n}^{2}=n^2+\alpha n+\beta;\quad n=0,1,2,...,$
then
$$[N,A^{\dagger}]=A^{\dagger},\;\;[N,A]=-A$$
and $$\;\;[A,A^{\dagger}]\Psi_n=2(2n-1+\alpha)\Psi_n\;\;\Longrightarrow\;\;[A,A^{\dagger}]=4N+2(\alpha-1)I\in\A(2).$$
Thereby the set of generators of the corresponding algebra is $\{ A^+, A, N, I\}$. Hence, as the dimension of $\mathcal{A}(1)$ is four, the dimension of $ \mathcal{A}(2)$ is also four. Thus the algebra $\A(2)$ is also isomorphic to the harmonic oscillator  algebra.\\
{\bf The case p=3:}
If $b_n^2=n^3 +\alpha n^2 +\beta n + \gamma,\:\;\;\forall n \in\mathbb N, $ then we have
$$[N,A^{\dagger}]=A^{\dagger},\;\;[N,A]=-A\quad\text{and}$$ $$[A,A^{\dagger}]=6N^2+2(2\alpha-3)N+2(1+\beta -\alpha)I.$$
Since $N$ and $I$ are elements of the algebra, $N^2$ is also in the algebra.
Let us prove that for any $m\in \mathbb N$, $(A^{\dagger})^m$ is in the algebra.
For, since $$[N^2,A^{\dagger}]=2A^{\dagger}N+A^{\dagger},\quad\text{and}\quad N^2, A^{\dagger}\in\A(3),$$ then $A^{\dagger}N$ is in the algebra.
Let us prove, by induction, that for any $m\in \mathbb N$, $$[{A^{\dagger}}^{m},A^{\dagger}N]=-m{A^{\dagger}}^{m+1}.$$
For $m=1$, we have 
\begin{eqnarray}
[{A^{\dagger}},A^{\dagger}N]&=& {A^{\dagger}}^{2}N-{A^{\dagger}}N{A^{\dagger}}\nonumber\\
&=&-{A^{\dagger}}[N,{A^{\dagger}}]=-{A^{\dagger}}^{2}.\nonumber
\end{eqnarray}
Therefore ${A^{\dagger}}^{2}\in\A(3)$.
Suppose that for $m\geq 1$, we have $$[{A^{\dagger}}^{m},A^{\dagger}N]=-m{A^{\dagger}}^{m+1} \Longrightarrow\;\;{A^{\dagger}}^{m+1}N-{A^{\dagger}}N{A^{\dagger}}^{m}=-m{A^{\dagger}}^{m+1}.$$
Now for $m+1$, we have
\begin{eqnarray}
[{A^{\dagger}}^{m+1},A^{\dagger}N]&=&{A^{\dagger}}^{m+2}N-{A^{\dagger}}N{A^{\dagger}}^{m+1}\nonumber\\
&=& {A^{\dagger}}^{m+2}N- {A^{\dagger}}^{m+1}N{A^{\dagger}} + {A^{\dagger}}^{m+1}N{A^{\dagger}}-    {A^{\dagger}}N{A^{\dagger}}^{m+1}\nonumber\\
&=& {A^{\dagger}}^{m+1}\left({A^{\dagger}}N-N{A^{\dagger}}\right)+ \left({A^{\dagger}}^{m+1}N-{A^{\dagger}}N{A^{\dagger}}^{m}\right){A^{\dagger}}\nonumber\\
&=& -{A^{\dagger}}^{m+2}-m{A^{\dagger}}^{m+1}{A^{\dagger}}\nonumber\\
&=& -(m+1){A^{\dagger}}^{m+2}.\nonumber
\end{eqnarray}
Therefore by induction $\{{A^{\dagger}}^{m}\;:\;m\geq 1\}\subset\mathcal A(3)$.  Thus, the algebra $\mathcal A(3)$  is of infinite dimension.\\
{\bf General case:} In general, let us prove that $\mathcal A(p)$ is of infinite dimension for $p\geq 3$. 
Let $b_n$ have the following form

\begin{equation}
b_n^2=Q(n)=\sum_{i=0}^{p}\alpha_in^i;\quad \alpha_i \in \mathbb R, \quad i=0,1,...,p.
\end{equation} 
It is easy to see that 
\begin{equation}\label{cr}
[A,A^{\dagger}]=\sum_{i=0}^{p-1}\theta_iN^i= \theta_{p-1}N^{p-1}+...+ \theta_{2}N^{2}+\theta_{1}N+\theta_{0}I,
\end{equation}
where $\theta_i \in \mathbb R,$ $i=0,1,...,p-1$.
Since $N,I\in\A(p)$,

\begin{equation}\label{q0}
\mathcal W^{0}=N^{p-1}+\gamma_{p-2}N^{p-2}+...+ \gamma_{2}N^{2}\in\A(p),
\end{equation} 
where $\gamma_{p-1}=1$, and $\gamma_i =\frac{\theta_{i}}{\theta_{p-1}};\quad i=2,...,p-2$. 
The following commutation relations can easily be computed:
\begin{equation}\label{q1}
[N^2,A^{\dagger}]=2A^{\dagger}N+A^{\dagger}
\end{equation}
\begin{equation}\label{q2}
[N^3,A^{\dagger}]=3A^{\dagger}N^2+3A^{\dagger}N+A^{\dagger}
\end{equation}
\begin{equation}\label{q3}
[N^4,A^{\dagger}]=4A^{\dagger}N^3+6A^{\dagger}N^2+4A^{\dagger}N+A^{\dagger}.
\end{equation}
That is, in general we have
\begin{equation}\label{q4}
[N^k,A^{\dagger}]= \sum_{i=1}^{k}c^i_kA^{\dagger}N^{k-i},\quad\text{where}~~c_k^i\in\mathbb{R}.
\end{equation}
Using (\ref{q0}) and (\ref{q4}), we have 
\begin{equation}\label{q10}
\mathcal W^{+}:= [\mathcal W^{0},A^{\dagger}]=\sum_{i=2}^{p-1}\sum_{j=1}^{i}\gamma_ic^{j}_{i} A^{\dagger}N^{i-j},
\end{equation}
Since $\mathcal W^{0}\in\A(p)$, we have $\mathcal W^{+} \in \mathcal{A}(p).$ 
After $(p-3)$-iterations it can be seen that
\begin{eqnarray}\label{rq1}
\mathcal W^+_{(p-2)}&:=&\left[A^{\dagger}...\left[A^{\dagger},\left[A^{\dagger}, \mathcal W^{+}\right] \right]... \right]\\
&=&(-1)^{p-1}(p-1)!A^{\dagger (p-2)}N\nonumber
+f(p)A^{\dagger (p-2)},
\end{eqnarray}
where $f(p)$ is some function of $p.$
Since $\mathcal W^+\in\A(p)$, we get $\mathcal W^+_{(p-2)}\in\A(p)$. 
Further, the following commutation relation can easily be verified by induction
\begin{equation}\label{cq2}
\left[A^{\dagger m},A^{\dagger (p-2)}N\right]= -mA^{\dagger (p-2+m)},\:\: m\geq 1.
\end{equation}
Now, (\ref{rq1}) and (\ref{cq2}) imply that
\begin{eqnarray}\label{cq3}
\left[A^{\dagger},\mathcal W^+_{(p-2)}\right]&=&(-1)^{p-1}(p-1)!\left[A^{\dagger},A^{\dagger (p-2)}N\right]\nonumber\\
&=&(-1)^{p}(p-1)!A^{\dagger (p-1)}.
\end{eqnarray}
Thereby, since $A^{\dagger}, \mathcal W^+_{(p-2)}\in\A(p)$, we see $A^{\dagger (p-1)}\in\mathcal{A}(p).$
Again using the relation (\ref{cq2}), we get
\begin{eqnarray}\label{cq4}
\left[A^{\dagger (p-1)},\mathcal W^+_{(p-2)}\right]&=&(-1)^{p-1}(p-1)!\left[A^{\dagger (p-1)},A^{\dagger (p-2)}N\right]\nonumber\\
&=& (-1)^{p}(p-1)!(p-1)A^{\dagger (2p-3)},\:\: m\geq 1;
\end{eqnarray}
Thereby, $A^{\dagger {(2p-3)}}  = A^{\dagger {(p-1)+(p-2)}}\in\mathcal{A}(p).$
By iteration, we can prove that $A^{\dagger {(p-1)+m(p-2)}}\in\mathcal{A}(p)$ for every $m\geq 1$. Further, for $p\geq 3$, the operators $A^{\dagger {(p-1)+m(p-2)}}$ are new elements of $\mathcal{A}(p)$ for every $m\geq 1$.
Thereby $\mathcal{A}(p)$ is of infinite dimension.

\subsection{Proof of Corollary \ref{c1}}
 The proof follows from the cases $p=0, 1, 2.$
\section{Some examples}
 \subsection{Hermite Polynomials:}
The Hermite polynomials  $H_n(x)$ are given by \cite{key14, key15, key16},
\begin{equation}\label{e8}
H_n(x)= n!\sum_{\nu=0}^{[\frac{n}{2}]}\frac{(-1)^\nu(2x)^{n-2\nu}}{\nu !(n-2\nu)!}.
\end{equation}
These polynomials are orthogonal in the Hilbert space
$\mathcal{H}= L^2\left(\mathbb R; \frac{1}{\sqrt{\pi}}\exp{(-x^2)}dx\right).$
The normalized polynomials $\{\Psi_n(x) \}_{n=0}^\infty$ takes the form
\begin{eqnarray}\label{e88}
\Psi_n(x)&=&\pi^\frac{1}{4} d_n^{-1}H_n(x),\quad \\
d_n&=& \left(2^n n!\sqrt{\pi}  \right)^{\frac{1}{2}},\quad n\geq 0.
\end{eqnarray}
The three term recurrence relations for the Hermite polynomials (for example see \cite{key14}), that is  the formula (\ref{Re}), is satisfied with 
\begin{eqnarray}
b_n=\frac{1}{2}\left(\frac{d_{n+1}}{d_n}  \right)= \sqrt{\frac{n+1}{2}}.
\end{eqnarray}
It is well know that the oscillator algebra generalized by this family of orthogonal polynomials is of finite dimension and we can see that $b_n^2$ has the form (\ref{T1}).
\subsection{The Legendre Polynomials:}
The Legendre polynomials are defined by 
\begin{equation}
P_n(x)={}_2F_1\left(-n, n+1;1;\frac{1-x}{2}\right).
\end{equation}
and these polynomials are orthogonal in the Hilbert space $\mathcal H=L^2([-1,1]; 2^{-1}dx)$.
The normalized polynomials $\{\Psi_n(x)_{n=0}^\infty  \}$ are given by
\begin{equation}
\Psi_n(x)=\sqrt{2}\Phi_n(x),\quad                 \Phi_n(x)=\sqrt{\frac{2n+1}{2}}P_n(x),\quad n\geq 0.
\end{equation}
Taking into account the recurrence relations for the Legendre polynomials (see \cite{key3}) we obtain the three term recurrence relation (\ref{Re}) with 
\begin{eqnarray}
b_n= \sqrt{\frac{(n+1)^2}{(2n+1)(2n+3)}},\quad n\geq 0.
\end{eqnarray}
In this case, $b_n^2$ does not take the form (\ref{T1}). Thereby the generalized oscillator algebra generated by the Legendre polynomials is of infinite dimension.
\subsection{The Gegenbauer Polynomials} These polynomials are orthogonal in the Hilbert space
$\mathcal H=L^2\left([-1,1]; (d_0(\alpha))^{-2}(1-x^2)^\alpha dx \right)$ 
where 
\begin{equation}
(d_0(\alpha))^{2}= 2^{2\alpha +1}\frac{(\Gamma(\alpha +1))^2}{(\Gamma(2\alpha +2))}.
\end{equation}
The ultraspherical polynomials are defined by the hypergeometric function \cite{AI}
\begin{equation}
P_n^{(\alpha, \alpha)}(x)=\frac{(\alpha +1)_n}{n!}{}_2F_1\left(-n, n+2\alpha +1;\alpha +1;\frac{1-x}{2}\right),
\end{equation}
where Pochhammer-symbol $(\beta)_n$ is defined by $(\beta)_0=1,$ and $(\beta)_n=\beta(\beta +1)(\beta +2)...(\beta +n -1);\: \: n\geq 1.$ For $\alpha >-1$ the following orthogonal relation is valid
$$\int_{-1}^1 P_n^{(\alpha, \alpha)}(x)P_m^{(\alpha, \alpha)}(x)(1-x^2)^{\alpha} dx=d_n^2\delta_{mn},\:\: n, m\geq 0, $$ with the constant of normalization $d_n$  given by 
$$d_n^2= \frac{2^{2\alpha +1}\left( \Gamma(n+2\alpha +1)  \right)^2}{(2n+2\alpha +1)n!\Gamma(n+\alpha +1)},\:\:n\geq 0.$$
For $\alpha=\lambda-2^{-1},\:\lambda >-2^{-1}$ and $n\geq 0$, the Gegenbauer polynomials are defined as \cite{key14}
\begin{equation*}
P_n^{(\lambda)}(x)=\frac{\Gamma(n+1)\Gamma(n+2\alpha +1)}{\Gamma(2\alpha +1)\Gamma(n+\alpha +1)}P_n^{(\alpha, \alpha)}(x).\nonumber
\end{equation*}
The normalized polynomials are given by \cite{keyb}
\begin{equation}
\Psi_n(x)=d_0d_n^{-1}P_n^{(\alpha, \alpha)}(x)
\end{equation}
The functions $\Psi_n(x)$ satisfy the relation (\ref{Re}) with 
\begin{equation}
b_n^2= \frac{(n+1)(n+2\alpha +1)}{(2n+2\alpha +1)(2n+2\alpha +3)};\quad n\geq 0,\:\: b_{-1}=0.
\end{equation}
It is clear that this $b_n^2$ is strictly different from (\ref{T1}). Thereby, the oscillator algebra associated with this family of orthogonal polynomials is of infinite dimension.
\section{Reduced algebras in multi-boson systems}
The dynamics of $(N+1)$-boson system is assumed to be governed by a Hamiltonian operator of the form \cite{RH,AH}:
\begin{eqnarray}\label{RH}
H &=& h_0(a_0^*a_0,...,a_N^*a_N) + g_0(a_0^*a_0,...,a_N^*a_N)a_0^{k_0}\cdots a_N^{k_N}
\\ &+& a_0^{-k_0}\cdots a_N^{-k_N}\overline{g}_0(a_0^*a_0,...,a_N^*a_N)\nonumber,
\end{eqnarray}
where $(a_0,...,a_N)$ and $(a_0^*,...,a_N^*)$ are bosonic annihilation and  creation operators respectively with standard Heisenberg commutation relations \cite{RH}.
The monomial $a_0^{k_0}\cdots a_N^{k_N},\:\: k_0,...,k_N \in {\mathbb Z}$  can be thought of as an operator which describes the subsequent creation and annihilation of the clusters of the bosonic modes.
The operator $g_0(a_0^*a_0,....,a_N^*a_N)$ is a kind of generalization of the coupling constant. The operator $h_0(a_0^*a_0,....,a_N^*a_N)$ can be chosen as a free Hamiltonian being a weighted sum of the occupation number operators of the elementary modes $a_0^*a_0,...,a_N^*a_N$.

The reduced algebra associated with this system in a reduced Hilbert space $\mathcal H_{\lambda_1,...,\lambda_N}$, which is formed by the orthonormal vectors $|\lambda_0,\lambda_1,...,\lambda_N\rangle$, is denoted by  $\A_{red}$. In order to analyze the quantum system described by the Hamiltonian (\ref{RH}), the following operators were introduced in \cite{RH}.
\begin{eqnarray}
A&:=&{g}_0(a_0^*a_0,....,a_N^*a_N)a_0^{k_0}...a_N^{k_N}\label{a},\\
A_i&=&A_i^{*}:=\sum_{j=0}^{N}\alpha_{ij}a_j^*a_j\label{b}.
\end{eqnarray}
and
\begin{eqnarray*}
A^*A&=&|{g}_0(a_0^*a_0-k_0,....,a_N^*a_N-k_N)|^2\mathcal{P}_{k_0}(a_0^*a_0-k_0)\cdots\mathcal{P}_{k_N}(a_N^*a_N-k_N),\\
AA^*&=&|{g}_0(a_0^*a_0,....,a_N^*a_N)|^2\mathcal{P}_{k_0}(a_0^*a_0)\cdots\mathcal{P}_{k_N}(a_N^*a_N),
\end{eqnarray*}
where
$$\mathcal{P}_k(a^*a)=a^ka^{-k}=\left\{\begin{array}{ccc}
a^k(a^*)^k=(a^*a+1)\cdots(a^*a+k)&\text{if}&k>0\\
1&\text{if}&k=0\\
(a^*)^{-k}a^{-k}=a^*a(a^*a-1)\cdots(a^*a-k+1)&\text{if}&k<0\end{array}.\right.$$
The operators $A_0, A$ and $A^*$ satisfy the commutation relations
$$[A_0,A]=-A,\quad [A_0,A^*]=A^*$$
and act on the basis vectors as
\begin{eqnarray}
A_0|\lambda_0,\lambda_1,...,\lambda_N\rangle&=&\lambda_0|\lambda_0,\lambda_1,...,\lambda_N\rangle\label{RAX}\\
A|\lambda_0,\lambda_1,...,\lambda_N\rangle&=&\sqrt{\mathcal G (\lambda_0-1,\lambda_1,...,\lambda_N)}|\lambda_0-1,\lambda_1,...,\lambda_N\rangle\label{RAY}\\
A^*|\lambda_0,\lambda_1,...,\lambda_N\rangle&=&\sqrt{\mathcal G (\lambda_0,\lambda_1,...,\lambda_N)}|\lambda_0+1,\lambda_1,...,\lambda_N\rangle.\label{RAZ}
\end{eqnarray}
Replacing the occupation numbers $a_0^*a_0,...,a_N^*a_N$ by the operators $A_0,A_1,$ $...,A_N$ one obtains
\begin{eqnarray}
A^*A&=&\mathcal{G}(A_0-1,A_1,...,A_N)\\
AA^*&=&\mathcal{G}(A_0,A_1,...,A_N),
\end{eqnarray}
where $\mathcal G $ is uniquely determined by $g_0$, the polynomials $\mathcal P_{k_0}, ...,\mathcal P_{k_N}$ and the linear map (\ref{b}). The reduced algebra $\A_{red}$ is generated by the operators  $A_0, A$, $A^*$ and $I$. The operators $A^*A$ and $AA^*$ are diagonal in the standard Fock basis,
$$|n_0,n_1,...,n_N\rangle=\frac{1}{\sqrt{n_0!\cdots n_N!}}(a_0^*)^{n_0}\cdots (a_N^*)^{n_N}|0\rangle.$$
The maximal system of commuting observables is diagonalized in the Fock basis and the eigenvalues of $A_0,A_1,...,A_N$ on $|n_0,n_1,...,n_N\rangle$ are given by
$$\lambda_i=\sum_{j=0}^{N}\alpha_{ij}n_j;\quad i=0,1,...,N.$$

In this section our aim is to give the necessary and sufficient form of the Hamiltonian(\ref{RH}) for the associated reduced algebra to be of finite dimension.
From the equations (\ref{RAX})-(\ref{RAZ}) one can see that the actions of $A_0, A$ and $A^*$ on the vector $|\lambda_0,\lambda_1,...,\lambda_N\rangle$  do not affect the parameters $\lambda_1,...,\lambda_N$. So by letting 
$b_{\lambda_0}=\sqrt{\mathcal G (\lambda_0,\lambda_1,...,\lambda_N)}$ and $\Psi_{\lambda_0}= |\lambda_0,\lambda_1,...,\lambda_N\rangle$ the  equations (\ref{RAX})-(\ref{RAZ}) can be rewritten as
\begin{eqnarray}\label{RAA}
A_0\Psi_{\lambda_0}&=&\lambda_0\Psi_{\lambda_0}\\
A\Psi_{\lambda_0}&=&b_{\lambda_0-1}\Psi_{\lambda_0-1}\\
A^*\Psi_{\lambda_0}&=&b_{\lambda_0}\Psi_{\lambda_0+1},
\end{eqnarray}
According to Theorem \ref{t1} the algebra generated by $\{I, A_0, A, A^*   \}$ is of  finite dimension if and only if  $b^2_{\lambda_0}$ is a polynomial of degree two in ${\lambda_0}$. But, since $b^2_{\lambda_0}$ depends on $\lambda_1,...,\lambda_{N}$, the coefficients of the polynomial must be  functions of $\lambda_1,...,\lambda_{N}$ . Hence, the necessary and sufficient condition for the reduced algebra to be of  finite dimension is  that the function $\mathcal G$ must have the following form: 
 \begin{equation}\label{RA-1}
b^2_{\lambda_0}=\mathcal G(\lambda_0,\lambda_1,...,\lambda_N)=\lambda_0^2\mathcal U_0(\lambda_1,...,\lambda_N) + \lambda_0\mathcal U_1(\lambda_1,...,\lambda_N)+\mathcal U_2(\lambda_1,...,\lambda_N),
\end{equation}
where $\mathcal U_i;~\:\:i=0,1,2$ are $N$ variables real valued functions.
Thereby, the complex valued function $g$ defined in (\ref{RH}) must satisfy the following condition (we consider it as a function of variables $x_1,...,x_N$)
\begin{eqnarray}\label{RC}
& &|g_0(x_0,...,x_N)|^2\mathcal P_{k_0}(x_0)...\mathcal P_{k_N}(x_N)\\
& &\quad=\lambda_0^2\mathcal U_0(\lambda_1,...,\lambda_N) + \lambda_0\mathcal U_1(\lambda_1,...,\lambda_N)+\mathcal U_2(\lambda_1,...,\lambda_N),\nonumber
\end{eqnarray}
where the polyomials $\mathcal P_{k}$ are defined as follows
\begin{equation}\label{d}
\mathcal P_{k}(x) = \left \{ \begin{array}{ll}
(x+1)...(x+k),\qquad \quad \:\:\mbox{if} \qquad k>0 \\
1\:\: \qquad \qquad \qquad\qquad\quad \quad\mbox{if}\qquad k=0\\
x(x-1)...(x-k+1),\quad \:\mbox{if}\qquad k<0 ,\\
\end{array}\right.
\end{equation}
and $$\lambda_i=\sum_{j=0}^N \alpha_{ij}x_j;~~i=0,1,2,...N.$$
As an example, let us work with the two mode Hamiltonian 
\begin{eqnarray}\label{RH2}
H &=& h_0(a_0^*a_0,a_1^*a_1) + g_0(a_0^*a_0,a_1^*a_1)a_0^{k_0}a_1^{k_1}
\\ & &\quad+ a_0^{-k_0}a_1^{-k_1}\overline{g}_0(a_0^*a_0,a_1^*a_1).\nonumber
\end{eqnarray}
For this Hamiltonian, the condition (\ref{RC}) is reduced to
\begin{eqnarray}\label{RC2}
|g_0(x_0,x_1)|^2\mathcal P_{k_0}(x_0)\mathcal P_{k_1}(x_1)&=&\left( \alpha_{0,0}x_0 + \alpha_{0,1}x_1\right)^2\mathcal U_0\left( \alpha_{1,0}x_0+ \alpha_{1,1}x_1\right)\nonumber\\ 
& &+\left( \alpha_{0,0}x_0 + \alpha_{0,1}x_1\right)\mathcal U_1\left(\alpha_{1,0}x_0+ \alpha_{1,1}x_1\right)\nonumber\\
& &\qquad+\mathcal U_2\left(\alpha_{1,0}x_0+ \alpha_{1,1}x_1\right).
\end{eqnarray}
Thereby, the necessary and sufficient condition for the reduction algebra to be isomorphic to the harmonic oscillator algebra is the following: 
\begin{eqnarray}\label{RC2O}
\mathcal U_0:=0.
\end{eqnarray}
In the case, $k_0,\:k_1\geq 0$, we have 
\begin{equation}\label{CA1}
|g_0(x_0,x_1)|^2=\frac{\left( \alpha_{0,0}x_0 + \alpha_{0,1}x_1\right)\cdot\mathcal U_1\left( \alpha_{1,0}x_0+ \alpha_{1,1}x_1\right)+\mathcal U_2\left(\alpha_{1,0}x_0+ \alpha_{1,1}x_1\right)}{\mathcal P_{k_0}(x_0)\mathcal P_{k_1}(x_1)}.
\end{equation}
Further, we suppose $g$ to be null function in the set where the expression (\ref{CA1}) is not defined.
For simplicity in (\ref{CA1}) we may take $\mathcal U_1=1$ and $\mathcal U_2=c$, some constant; then the necessary and sufficient condition for the reduction algebra to be isomorphic to the harmonic oscillator algebra becomes
\begin{equation}\label{CA2}
|g_0(x_0,x_1)|^2=\frac{\left( \alpha_{0,0}x_0 + \alpha_{0,1}x_1\right)+c}{\mathcal P_{k_0}(x_0)\mathcal P_{k_1}(x_1)},
\end{equation}
 which implies that 
\begin{equation}\label{CA3}
g_0(x,y)=\exp(i\theta)\left[\frac{\left( \alpha_{0,0}x + \alpha_{0,1}y\right)+c}{\mathcal P_{k_0}(x)\mathcal P_{k_1}(y)}\right]^{\frac{1}{2}}.
\end{equation}
Now let us look at a particular example given in \cite{AH}, where 
\begin{eqnarray}\label{RH22}
H &=& h_0(a_0^*a_0,a_1^*a_1) + g_0(a_0^*a_0,a_1^*a_1)a_0^{k_0}a_1^{*k_1}
\\ & &\qquad + a_0^{*k_0}a_1^{k_1}\overline{g}_0(a_0^*a_0,a_1^*a_1),\nonumber
\end{eqnarray}
with the matrix elements 
\begin{equation}
\alpha_{0,0}= \frac{1}{k_0},\quad \alpha_{0,1}= 0,\quad \alpha_{1,0}= k_1, \quad \alpha_{1,1}= k_0.
\end{equation}
In this case, according to (\ref{CA3}), the necessary and sufficient condition for the algebra, $\A_{red}$ to be isomorphic to the harmonic oscillator algebra is that the function $g$ must have the following form:
\begin{equation}\label{CA5}
g_0(x,y)=\exp(i\theta)\left[\frac{x+k_0c}{k_0\prod_{i=1}^{k_0}(x+i)\prod_{j=0}^{k_1-1}(y-j)}\right]^{\frac{1}{2}}.
\end{equation}
\section{conclusion}
In this paper, the case of polynomials that are orthogonal with respect to a symmetric measure satisfying three terms recurrence relations has been studied. The case of more general polynomials will be considered in the future. Further, the normalized 2D-Hermite polynomials $H_{n,m}(z,\overline{z})$ satisfy the three term recurrence relation \cite{Wu1, GA}
\begin{equation}\label{He1}
zH_{m,n}(z,\overline{z})=\sqrt{m+1}H_{m+1,n}(z,\overline{z})+\sqrt{n}H_{m,n-1}(z,\overline{z}).
\end{equation}
The quaternionic extension of these 2D Hermite polynomials, introduced in \cite{Thi1, Thi2}, $H_{m,n}({\bf{q}},\overline{{\bf{q}}})$ also satisfy the same recurrence relation with the complex number $z$ replaced by a quaternion ${\bf{q}}$ in (\ref{He1}). The 2D-Zernike polynomials also satisfy a three term recurrence relation of the form \cite{Wu2}
$$zP_{m,n}^{\alpha}(z,\overline{z})=a_{m,n}P_{m+1,n}^{\alpha}(z,\overline{z})
+b_{m,n}P_{m,n-1}^{\alpha}(z,\overline{z}).$$
The recurrence relations of these 2D polynomials do not apply to (\ref{Re}), and thereby the theory developed in this note does not fit these polynomials. In this regard, for these algebras and their dimension-wise classification, a separate theory needs to be established for the 2D polynomials. We shall consider these issues in our future work.
\begin{center}
{\bf{Acknowledgements}}
\end{center}
The authors would like to thank S. Twareque Ali for suggesting the problem of this note and for the valuable discussions on the early stage of this manuscript. We would also like to thank an anonymous referee for his valuable comments.

\end{document}